\documentclass[10pt, conference, a4paper]{IEEEtran}

\usepackage{amsmath}
\usepackage{amssymb}
\usepackage{amsthm}
\usepackage{color}
\usepackage{colortbl}
\usepackage{verbatim}
\usepackage[pdftex]{graphicx}
\usepackage{url}
\usepackage{cite}
\usepackage{tikz}
\usetikzlibrary{matrix}
\usepackage{hyperref}

\usepackage{slashbox}

\newtheorem{theorem}{Theorem}
\newtheorem{lemma}{Lemma}

\newtheorem{corollary}{Corollary}
\newtheorem{definition}{Definition}
\newtheorem{example}{Example}
\newtheorem{remark}{Remark}
\newtheorem{construction}{Construction}

\newcommand{\field}[1]{\mathbb{#1}}

\newcommand{\F}{\field{F}}
\newcommand{\Z}{\field{Z}}

\newcommand{\cA}{{\mathcal A}}

\newcommand{\cC}{{\mathcal C}}

\newcommand{\cG}{{\mathcal G}}

\newcommand{\cP}{{\mathcal P}}

\newcommand{\sP}{\cP}
\newcommand{\sG}{\cG}

\newcommand{\Gr}{\smash{{\sG\kern-1.5pt}_q\kern-0.5pt(n,k)}}
\newcommand{\Gk}{\smash{{\sG\kern-1.5pt}_q\kern-0.5pt(n,k_1)}}
\newcommand{\Gkk}{\smash{{\sG\kern-1.5pt}_q\kern-0.5pt(n,k_2)}}
\newcommand{\Grtwo}{\smash{{\sG\kern-1.5pt}_2\kern-0.5pt(n,k)}}
\newcommand{\Gkone}{\smash{{\sG\kern-1.5pt}_q\kern-0.5pt(n,k_1)}}
\newcommand{\Gktwo}{\smash{{\sG\kern-1.5pt}_q\kern-0.5pt(n,k_2)}}
\newcommand{\Ps}{\smash{{\sP\kern-2.0pt}_q\kern-0.5pt(n)}}

\newcommand{\LIN}[3]{\ensuremath{[#1,#2,#3]}}
\newcommand{\interval}[1]{\ensuremath{[#1)}}
\newcommand{\simplex}[1]{\ensuremath{\mathcal{S}_{#1}}}
\newcommand{\anti}[2]{\ensuremath{\mathcal{A}_{#1,#2}}}
\newcommand{\tn}[1]{\textnormal{#1}}

\begin{document}

\title{Optimal Binary Locally Repairable\\ Codes via Anticodes}

\IEEEoverridecommandlockouts

\author{\IEEEauthorblockN{Natalia Silberstein and Alexander Zeh}\thanks{N. Silberstein has been supported in part at the Technion by a Fine Fellowship. A. Zeh has been supported by the German research council (Deutsche Forschungsgemeinschaft, DFG) under grant Ze1016/1-1.}
\IEEEauthorblockA{
Computer Science Department\\
Technion---Israel Institute of Technology\\
Haifa 32000, Israel\\
\texttt{natalys@cs.technion.ac.il, alex@codingtheory.eu}
}}
\maketitle

\begin{abstract}
This paper presents a construction for several families of optimal binary locally repairable codes (LRCs) with  small locality (2 and 3). This construction is  based on various anticodes. It provides binary LRCs which attain the Cadambe--Mazumdar bound. Moreover, most of these codes are optimal with respect to the Griesmer bound.
\end{abstract}



\section{Introduction}
Locally repairable codes (LRCs) are a family of erasure  codes which allow local correction of erasures, where any code symbol can be recovered by using a small (fixed) number of other code symbols. The concept of LRCs was motivated by application to distributed storage systems (DSSs), (see e.g.~\cite{dimakis_network_2010, dimakis_survey_2011} and references therein).
DSSs store data across a network of nodes in a redundant form to ensure resilience against node failures. Using of LRCs to store data in DSSs enables to repair a failed node locally, i.e., by accessing  a small number of other nodes in the system.

The $i$th code symbol $c_i$, $1\leq i\leq n$, of an $\LIN{n}{k}{d}$ linear code~$\cC$  is said to have \emph{locality} $r$ if $c_i$ can be recovered by accessing at most $r$ other code symbols. A code $\cC$ is said to have locality $r$ if all its symbols have locality $r$. Such codes are referred to as locally repairable (or recoverable) codes (LRCs).
LRCs were introduced by Gopalan~\textit{et al.} in~\cite{gopalan_locality_2012}. It was shown in~\cite{gopalan_locality_2012} that the minimum distance  of an $[n,k,d]$ LRC with locality $r$ should satisfy the following  generalization of the Singleton bound
\begin{equation} \label{eq_GeneralizedSingleton}
d \leq n-k+2-\left \lceil \frac{k}{r} \right \rceil.
\end{equation}
Constructions of LRCs which attain this bound were presented in~\cite{westerback_almost_2014,gopalan_explicit_2014,gopalan_locality_2012,tamo_family_2014}. Further generalizations of LRCs to the codes which can locally correct more than one erasure,
 to the LRCs with multiple repair alternatives,
 and to the vector LRCs were considered in~\cite{kamath_codes_2014, kamath_explicit_2013,pamies-juarez_locally_2013,papailiopoulos_locally_2014,prakash_optimal_2012,prakash_codes_2014, rawat_optimal_2014,rawat_locality_2014, silberstein_optimal_2013,tamo_family_2014}.

However, to attain the bound~(\ref{eq_GeneralizedSingleton}) and its generalizations~\cite{kamath_codes_2014,papailiopoulos_locally_2014,prakash_optimal_2012,rawat_locality_2014}, the known codes should be defined over a large finite field. In~\cite{tamo_family_2014} Tamo and Barg presented LRCs satisfying the bound~(\ref{eq_GeneralizedSingleton}), which are defined over a field of size slightly greater than the length of the code. This is the smallest field size for optimal LRCs known so far.

Codes over small (especially binary) alphabets are of a particular interest due to their implementation ease. Recently, a new bound for LRCs which takes the size of the alphabet into account was established by Cadambe and Mazumdar~\cite[Thm. 1]{cadambe_upper_2013}. They showed that
the dimension $k$ of an \LIN{n}{k}{d} LRC $\cC$ over $\F_q$ with locality $r$  is upper bounded by
\begin{equation} \label{eq_CMBound}
k \leq \min_{t \in \Z^{+}} \left\{tr + k_{opt}^{(q)}\big( n-t(r+1), d \big)  \right\},
\end{equation}
where $k_{opt}^{(q)}(n, d)$ is the largest possible dimension of a code of length $n$, for a given alphabet size $q$ and a given minimum distance $d$.
This bound applies to both linear and nonlinear codes. In the case of a nonlinear code, the parameter $k$ is defined as $|\cC|/ \log q$. Moreover, it was shown in~\cite{cadambe_upper_2013} that the family of binary simplex codes attains the bound~(\ref{eq_CMBound}) for $r=2$.
To the best of our knowledge, there are only three additional works that consider constructions of binary LRCs, namely the papers of Shahabinejad~\textit{et al.}~\cite{shahabinejad_efficient_2014}, Goparaju and Calderbank~\cite{goparaju_binary_2014}, and Zeh and Yaakobi~\cite{zeh_optimal_2015}, where the constructions of \cite{goparaju_binary_2014} and \cite{zeh_optimal_2015} consider binary cyclic LRCs.

In this paper we propose constructions of new binary LRCs which attain the bound~(\ref{eq_CMBound}). All our LRCs have a small locality ($r=2$ and $r=3$), moreover, most of our codes attain the Griesmer bound. Our constructions use a method of Farrell~\cite{farrell_linear_1970} based on anticodes. In particular, we modify a binary simplex code by deleting certain columns from its generator matrix. These deleted columns form an anticode. We investigate the properties of anticodes which allow constructions of LRCs with small locality. Also, we present optimal binary LRCs  with locality 2 based on subspace codes.

The rest of the paper is organized as follows.
In Section~\ref{sec:Preliminaries} we provide the necessary definitions, in particular, we define anticodes and describe a method of constructing a new code based on a simplex code and an anticode.
In Section~\ref{sec:Constructions} we present our constructions based on various choices of anticodes.
Conclusion is given in Section~\ref{sec:conclusion}.


\section{Preliminaries}
\label{sec:Preliminaries}

Let $\cC$ be a \emph{linear} $\LIN{n}{k}{d}$ \emph{code} of length $n$, dimension $k$ and minimum Hamming distance $d$ over $\F_q$. We say that a $k\times n$ generator matrix $G$ of $\cC$ is in a \emph{standard form} if it has the form $G=(I_k|A)$, where $I_j$ is an identity matrix of order~$j$ and $A$ is a $k\times (n-k)$ matrix. If $G$ is in a standard form then an $(n-k)\times n$ parity-check matrix $H$ can be easily obtained from $G$ in the following way: $H=(-A^{T}|I_{n-k})$. Note that for a binary code, we have $H=(A^T|I_{n-k})$. In the sequel, we will consider only binary codes.
The following simple lemma shows a sufficient condition on a parity-check matrix of a code with locality $r$.

\begin{lemma} \label{rm:locality from H}
An $[n,k,d]$ linear code has locality $r$ if  for every coordinate $i$, $1\leq i\leq n$,  there exists a row $R_i$ of weight at most $r+1$ in its parity-check matrix, which has a nonzero entry in the $i$th coordinate. In this case we say that the coordinate~$i$ \emph{is covered} by the row $R_i$.
\end{lemma}
The following two bounds will be used in the sequel.
The Plotkin bound (Thm.~\ref{thm:Plotkin}) holds for nonlinear codes, while the Griesmer bound is restricted to linear codes (Thm.~\ref{thm:Griesmer}).
\begin{theorem}[Plotkin Bound~{\cite[p. 43]{macwilliams_theory_1988}}] \label{thm:Plotkin}
Let $A_2(n,d)$ denote the largest number of  codewords in a binary code of length $n$ and minimum distance $d$.
 If $d$ is even and $2d>n$ then
$$A_2(n,d)\leq 2\left\lfloor\frac{d}{2d-n}\right\rfloor.
$$
\end{theorem}

\begin{theorem}[Griesmer Bound~{\cite[p. 547]{macwilliams_theory_1988}}] \label{thm:Griesmer}
The length $n$ of a binary linear code with dimension $k$ and minimum distance~$d$ must satisfy
\begin{equation*}
n \geq \sum_{i=0}^{k-1}\left\lceil\frac{d}{2^i}\right\rceil.
\end{equation*}
\end{theorem}

In the remaining part of this section we define an anticode and recall the anticode-based construction of binary linear codes by Farrell~\cite{farrell_linear_1970} (see also~\cite[p. 548]{macwilliams_theory_1988}).
An \emph{anticode} is a code which may contain repeated codewords and which has an \emph{upper} bound on the distance between the codewords. More precisely,
a binary linear anticode $\cA$ of length $n$ and maximum distance $\delta$ is a set of codewords in $\F_2^n$ such that the  Hamming distance between any pair of codewords is less than or equal to $\delta$. The generator matrix $G_{\cA}$  of $\cA$ is a $k \times n$ binary matrix such that all the $2^k$ combinations of its rows form the codewords of the anticode.
If $\text{rank}(G_A) = \gamma$, then each codeword occurs $2^{k-\gamma}$ times in the anticode.
Due to linearity, we have
\begin{equation}
\delta = \max_{\mathbf{a} \in \cA} \text{wt}(\mathbf{a}),
\end{equation}
where $\text{wt}(\mathbf{v})$ denotes the Hamming weight of a vector $\mathbf{v}$.

\begin{example} \label{ex_Anticode}
Let $G_{\cA}$ be a $3\times 3$ generator matrix given by
\begin{footnotesize}
\begin{equation*}
G_{\cA} = \begin{pmatrix}
1 & 1 & 0 \\
1 & 0 & 1 \\
0 & 1 & 1
\end{pmatrix}.
\end{equation*}
\end{footnotesize}
\hspace{-0.1cm}It generates  a binary linear anticode $\cA$ of length $n=3$ and $\delta=2$, where
the set of $2^3$ codewords is
\begin{equation*}
\cA = \big\{ (000),(110),(101),(011),(011),(101),(110),(000) \big\}.
\end{equation*}
\end{example}

The construction of Farrell~\cite{farrell_linear_1970}
which we use to construct optimal LRCs is based on a modification of a generator matrix for a binary simplex code.

A binary \emph{simplex} code \simplex{m} is a $\LIN{2^m-1}{m}{2^{m-1}}$ code with generator matrix $G_m$ whose columns consist of all distinct nonzero vectors in $\F_{2}^m$. In the rest of the paper we assume w.l.o.g. that $G_m$ is in the standard form and that the columns of $G_m$ are ordered according to their Hamming weight.

\begin{construction}[Farrell Construction~\cite{farrell_linear_1970}] \label{const_Farrell}
Let $G_m$ be the ${m \times (2^m-1)}$ generator matrix of a binary simplex code \simplex{m} and let $G_{\cA}$ be the $k \times n$ generator matrix with distinct columns of a binary linear anticode $\cA$ of length $n$ and  maximum distance~$\delta$. Then, the $m \times (2^m-1-n)$ matrix obtained by deleting the $n$ columns of $G_{\cA}$ from $G_m$ is a generator matrix of a binary $\LIN{2^m-1-n}{\leq m}{2^{m-1}-\delta}$ code.
\end{construction}

\begin{example} \label{ex_FarrelConstruction}
Let $G_4$ be the $4 \times 15$ generator matrix of a simplex code \simplex{4} and let $G_{\cA}$ be  the generator matrix of the anticode given in Example~\ref{ex_Anticode} with the additional first row of zeros. By deleting the columns of $G_{\cA}$ from $G_4$ we obtain the following matrix
\begin{tikzpicture}
\tikzset{BarreStyle/.style = {opacity=.5,line width=4 mm,line cap=round,color=#1}}
\tikzstyle{every node}=[font=\footnotesize]
\node(Z){$G_4\setminus G_{\cA}=$};
\matrix (A) [matrix of math nodes,column sep=0 mm, left delimiter={(},right delimiter={)}, right=10pt] at (Z.east)
{
1 & 0 & 0 & 0 & 1 & 1 & 1 & 0 & 0 & 0 & 1 & 1 & 1 & 0 & 1 \\
0 & 1 & 0 & 0 & 1 & 0 & 0 & 1 & 1 & 0 & 1 & 1 & 0 & 1 & 1 \\
0 & 0 & 1 & 0 & 0 & 1 & 0 & 1 & 0 & 1 & 1 & 0 & 1 & 1 & 1 \\
0 & 0 & 0 & 1 & 0 & 0 & 1 & 0 & 1 & 1 & 0 & 1 & 1 & 1 & 1 \\
};
\draw [BarreStyle=black] (A-1-8.north) to (A-4-8.south) ;
\draw [BarreStyle=black] (A-1-9.north) to (A-4-9.south) ;
\draw [BarreStyle=black] (A-1-10.north) to (A-4-10.south) ;
\end{tikzpicture}\\
where the shadowed columns in $\{8,9,10\}$ are deleted, which generates a  $\LIN{12}{4}{6}$ code. Note that this code attains the Griesmer bound.


\end{example}

\section{Constructions of Optimal Binary LRCs}
\label{sec:Constructions}

In this section we provide constructions of binary LRCs based on the Farrell construction (see Construction~\ref{const_Farrell}), by using various anticodes. We prove that our codes have a small locality ($r=2$ or $r=3$) and attain the Cadambe--Mazumdar bound~(\ref{eq_CMBound}). Most of our codes also attain the Griesmer bound (see Thm.~\ref{thm:Griesmer}).

 \subsection{LRCs based on Anticodes}
First, we generalize Example~\ref{ex_Anticode} and consider an anticode such that  all the vectors of length $s$ and weight 2 form the columns of its generator matrix. We denote such an anticode by $\cA_{s,2}$. For example, the anticode from Example~\ref{ex_Anticode} is an $\cA_{3,2}$ anticode.
 First, we need the following theorem about the parameters of~$\cA_{s,2}$.

 \begin{theorem}
 \label{thm:anticode-2weight}
  Let $\cA_{s,2}$ be a binary anticode such that all weight-2 vectors of length $s$ form the columns
 of its generator matrix~$G_{\cA}$.  Then $\cA_{s,2}$ has length $\binom{s}{2}$ and maximum weight ${\delta=\left\lfloor\frac{s^2}{4}\right\rfloor}$.
 \end{theorem}

\begin{proof}
There are $\binom{s}{2}$ vectors of weight 2 and length $s$, hence the number of columns in $G_{\cA}$ and the length of  $\cA_{s,2}$ is $\binom{s}{2}$.
Next, we prove that the value of the maximum weight~$\delta$ is $\left\lfloor\frac{s^2}{4}\right\rfloor$. Note that the $s\times \binom{s}{2}$ generator matrix $G_{\cA}$ is also an incidence matrix of a complete graph $K_s=(V,E)$ with $|V|=s$ and ${|E|=\binom{s}{2}}$. Therefore, the maximum weight $\delta$ of the anticode can be described in terms of maximum cut between a subset of vertices $S\subseteq V$  and its complement $S^c$, more precisely,
$$\delta=\max_{1\leq i\leq s}|\textmd{Cut}(S_i,S_i^c)|,
$$
where $S_i$ is a subset of $V$ of size $i$ and $S_i^c$ its complement of size $s-i$.
Since $K_s$ is an $(s-1)$-regular graph, it holds that for all $1\leq i\leq s$,
$(s-1)i=|\textmd{Cut}(S_i,S_i^c)|+2|E_i|$,
where $E_i\subseteq E$ is the set of edges between the vertices in $S_i$. Note that the induced subgraph $(S_i, E_i)$ of $K_s$ is a complete graph $K_i$ and then $|E_i|=\binom{i}{2}$. Thus,
$$|\textmd{Cut}(S_i,S_i^c)|=(s-1)i-i(i-1)=i(s-i)
$$
and
$$
\delta=\max_{1\leq i\leq s}\{i(s-i)\}=\left\lfloor\frac{s^2}{4}\right\rfloor.
$$
\end{proof}
As a consequence of Thm.~\ref{thm:anticode-2weight} and the Farrell construction we have the following theorem.

\begin{theorem} \label{thm:anticode-weight2}
Let $\simplex{m}$ be a $[2^m-1,m,2^{m-1}]$ simplex code, ${m\geq 4}$, and let $G_m$ be its generator matrix. Let $\cA_{s,2}$, ${s\leq m}$, be an anticode defined in Thm.~\ref{thm:anticode-2weight} and let $G_{\cA}$ be its generator matrix. We prepend $m-s$ zeros to every column of $G_{\cA}$ to form an $m\times \binom{s}{2}$ matrix whose columns are deleted from $G_m$ to obtain a generator matrix $G$ for a new code $\cC_{m,s,2}$.
Then $\cC_{m,s,2}$ is a $[2^m-\binom{s}{2}-1,m,2^{m-1}-\left\lfloor\frac{s^2}{4}\right\rfloor]$ LRC with locality $r=2$.

\end{theorem}

\begin{proof}
The prepending of zeros to $G_{\cA}$ does not change the length and the maximum weight of the anticode. Hence, the length, the dimension, and the minimum distance of the obtained code $\cC_{m,s,2}$ directly follow from the Farrell construction and Thm.~\ref{thm:anticode-2weight}.

To prove that the locality of $\cC_{m,s,2}$ is $r=2$, by Lemma~\ref{rm:locality from H}  we need  to show that every coordinate is covered by a row of weight 3 of the parity-check matrix for $\cC_{m,s,2}$ (note that clearly locality is not 1).
Since $G_m$ is  in the standard form, the generator matrix $G$ of $\cC_{m,s,2}$ is also in the standard
form. We denote by $G^i$, $1\leq i\leq m$, the submatrix of $G$ which consists of the set of columns of $G$ of weight $i$. Then the parity-check matrix $H$ of $\cC_{m,s,2}$ has the following form:
$$H=\left(
    \begin{array}{c|c|c|c|c}
        (G^2)^T & I_{\binom{m}{2}-\binom{s}{2}} & &  &  \\\hline
       \vdots &  &\ddots  & &  \\\hline
       (G^{m-1})^T &  &&  I_{m} &  \\\hline
      (G^{m})^T &  & &  & 1 \\
    \end{array}
  \right).
$$
We will show that by a simple modification of $H$ with  elementary operations on its rows we obtain a parity-check matrix $H'$ such that every coordinate of the code will be covered by a row of $H'$ of weight 3.
We define a partition of columns of $H$ into the parts $\{H^1,\ldots,H^m\}$ as follows.
 The part $H^1$ contains the first $m$ columns, and the part $H^i$, $2\leq i\leq m$, corresponds to the columns which contain $I_{\binom{m}{i}}$ included in the rows which contain $(G^i)^T$. (Note that if $s=m$ then $H^2=\emptyset$.)
Let consider the $x$th coordinate in $H^i$, $3\leq i\leq m-1$. There exists a row $R^{i}_{x}$ in $H$ with nonzero entry in this coordinate. This row $R^{i}_{x}$ also contains $i$ nonzero entries in the first $m$ coordinates. Let $R^{i+1}_{y}$ be a row in $H$ such that its $i+1$ nonzero entries in the first $m$ coordinates contain the first $i$ nonzero entries of $R^{i}_{x}$, and which also has one in the $y$th coordinate of $H^{i+1}$.
Then the coordinates $x$ in $H^i$ and $y$ in $H^{i+1}$ are covered by $R^{i}_{x}+R^{i+1}_{y}$, the weight-3 row of $H'$. Note that for every
$x$ in $H^i$ there is such $y$ in $H^{i+1}$. To show that any coordinate in $H^1\cup H^m$ has locality~2, note that when
 we add the last row of $H$ to each one of the $m$ rows which contain the rows of $(G^{m-1})^T$, we obtain $m$ rows of weight 3 in $H'$ with the first nonzero entry in the first $m$ coordinates and the last nonzero entry in the last coordinate. Thus, the obtained parity-check matrix $H'$ contains weight-3 rows which cover all the coordinates (and the last row of $H$, of weight $m+1$).

%
%
%
\end{proof}

In the following example we illustrate the idea of modification of  a parity-check matrix described in the proof of Thm.~\ref{thm:anticode-weight2}.
\begin{example} \label{ex_locality_reduction}
We consider the parity-check matrix $H$ of the $[12,4,6]$ code of Example~\ref{ex_FarrelConstruction} based on the anticode $\anti{3}{2}$. It has the following form, where the vertical lines show the partition of its columns into the parts $H^1,H^2,H^3,H^4$.

\begin{footnotesize}
\begin{equation*}
H=
\begin{array}{c}
    \left(\begin{tabular}{c|c|c|c}
    \tn{1 1 0 0} & \tn{1 0 0} & \tn{0 0 0 0} & \tn{0}\\
    \tn{1 0 1 0} & \tn{0 1 0} & \tn{0 0 0 0} & \tn{0}\\
    \tn{1 0 0 1} & \tn{0 0 1} & \tn{0 0 0 0} & \tn{0}\\
    \tn{1 1 1 0} & \tn{0 0 0} & \tn{1 0 0 0} & \tn{0}\\
    \tn{1 1 0 1} & \tn{0 0 0} & \tn{0 1 0 0} & \tn{0}\\
    \tn{1 0 1 1} & \tn{0 0 0} & \tn{0 0 1 0} & \tn{0}\\
    \tn{0 1 1 1} & \tn{0 0 0} & \tn{0 0 0 1} & \tn{0}\\
    \tn{1 1 1 1} & \tn{0 0 0} & \tn{0 0 0 0} & \tn{1}\\
        \end{tabular}
\right)
\end{array}
\end{equation*}
\end{footnotesize}
We add to each one of the $i$th rows of $H$, $4\leq i\leq 7$, the last row and obtain the following parity-check matrix:

 \begin{footnotesize}
$$H'=\begin{array}{c}
    \left(\begin{tabular}{c|c|c|c}
    \tn{1 1 0 0} & \tn{1 0 0} & \tn{0 0 0 0} & \tn{0}\\
    \tn{1 0 1 0} & \tn{0 1 0} & \tn{0 0 0 0} & \tn{0}\\
    \tn{1 0 0 1} & \tn{0 0 1} & \tn{0 0 0 0} & \tn{0}\\
    \tn{0 0 0 1} & \tn{0 0 0} & \tn{1 0 0 0} & \tn{1}\\
    \tn{0 0 1 0} & \tn{0 0 0} & \tn{0 1 0 0} & \tn{1}\\
    \tn{0 1 0 0} & \tn{0 0 0} & \tn{0 0 1 0} & \tn{1}\\
    \tn{1 0 0 0} & \tn{0 0 0} & \tn{0 0 0 1} & \tn{1}\\
    \tn{1 1 1 1} & \tn{0 0 0} & \tn{0 0 0 0} & \tn{1}\\
        \end{tabular}
\right),
\end{array}$$
\end{footnotesize}
such that every coordinate is covered by a weight-3 row.
\end{example}

\begin{corollary}
For $3\leq s\leq 5$, the code $\cC_{m,s,2}$ obtained in Thm.~\ref{thm:anticode-weight2} attains the bound~(\ref{eq_CMBound}). More precisely,
\begin{itemize}
  \item The code $\cC_{m,3,2}$ obtained by using the anticode $\cA_{3,2}$ is a $[2^m-4,m,2^{m-1}-2]$ LRC with locality $r=2$ which attains the bound~(\ref{eq_CMBound}).
  \item The code $\cC_{m,4,2}$ obtained by using the anticode $\cA_{4,2}$ is a $[2^m-7,m,2^{m-1}-4]$ LRC with locality $r=2$ which attains the bound~(\ref{eq_CMBound}).
  \item The code $\cC_{m,5,2}$ obtained by using the anticode $\cA_{5,2}$  is a $[2^m-11,m,2^{m-1}-6]$ LRC with locality $r=2$ which attains the bound~(\ref{eq_CMBound}).
\end{itemize}
\end{corollary}

\begin{proof}
To prove the optimality of the proposed codes we apply  the bound~(\ref{eq_CMBound}) with $t=1$ and use the Plotkin bound (see Thm.~\ref{thm:Plotkin}):

For $s=3$ we have $2+k_{opt}^{(2)}(2^m-7,2^{m-1}-2)\leq 2+\left\lfloor\log 2\left\lfloor\frac{2^{m-1}-2}{3}\right\rfloor\right\rfloor\leq 2+\left\lfloor\log(2^{m-1}-2)\right\rfloor=2+m-2=m$.

For $s=4$ we have $2+k_{opt}^{(2)}(2^m-10,2^{m-1}-4)\leq 2+\left\lfloor\log 2\left\lfloor\frac{2^{m-1}-4}{2}\right\rfloor\right\rfloor= 2+\left\lfloor\log(2^{m-1}-4)\right\rfloor=2+m-2=m$.

For $s=5$ we have $2+k_{opt}^{(2)}(2^m-14,2^{m-1}-6)\leq 2+\left\lfloor\log 2\left\lfloor\frac{2^{m-1}-6}{2}\right\rfloor\right\rfloor= 2+\left\lfloor\log(2^{m-1}-6)\right\rfloor=2+m-2=~m$.
\end{proof}

\begin{remark}
One can check that the codes $\cC_{m,3,2}$ and $\cC_{m,5,2}$ attain the Griesmer bound.
\end{remark}


Note that to apply our modification of a parity-check matrix in the proof for locality 2, the generator matrix of a code obtained by the Farrell construction should contain columns of consecutive weights. Based on this observation, we propose
a generalization of the previous construction of an anticode
and prove that the LRCs obtained from this anticode have locality $r=2$ and attain the Griesmer bound.

\begin{theorem} \label{thm:anticodeSm-1}
Let $\simplex{m}$ be a $[2^m-1,m,2^{m-1}]$ simplex code, $m\geq 4$, and let $G_m$ be its generator matrix. Let $\cA_{t;2,3,\ldots, t-1}$, $3\leq t\leq m$, be an anticode such that its generator matrix $G_{\cA}$ consists of all the columns in $\F_2^t$ of weights in ${\{2,3,\ldots, t-1\}}$.
 We prepend $m-t$ zeros to every column of $G_{\cA}$ to form the $m\times \sum_{i=2}^{t-1}\binom{t}{i}$ matrix whose columns are deleted from $G_m$ to obtain a generator matrix $G$ for a new code  $\cC_{m,t}$. Then  $\cC_{m,t}$ is
 a $[2^{m}-2^t+t+1,m,2^{m-1}-2^{t-1}+2]$ LRC with locality $r=2$ which attains the Griesmer bound.
\end{theorem}
\begin{proof}
First we prove that the anticode $\cA_{t;2,3,\ldots, t-1}$ has length $2^t-t-2$ and maximum weight $2^{t-1}-2$.
Note that the generator matrix of $\cA_{t;2,3,\ldots, t-1}$ can be obtained from the generator matrix $G_t$ of  the simplex code $\simplex{t}$ by removing $t$ columns of weight 1 and one column of weight $t$, and hence the length of $\cA_{t;2,3,\ldots, t-1}$ is $2^t-t-2$. Since all the codewords in $\simplex{t}$ have weight $2^{t-1}$ and from each row of $G_t$ we removed two ones to obtain a generator matrix for our anticode, where all the rows in $G_t$ have one of the removed ones in the same place, it follows that the maximum weight of $\cA_{t;2,3,\ldots, t-1}$ is $\delta=2^{t-1}-2$. Hence, since prepending zero rows to $G_{\cA}$ does not change the length and the maximum weight of the anticode,  $\cC_{m,t}$ is a $[2^{m}-2^t+t+1,m,2^{m-1}-2^{t-1}+2]$ code. The proof of locality is similar to the proof in Thm.~\ref{thm:anticode-2weight}.

To prove that $\cC_{m,t}$ attains the Griesmer bound we have
$$\sum_{i=0}^{m-1} \left \lceil \frac{2^{m-1}-2^{t-1}+2}{2^i} \right \rceil =\sum_{i=0}^{m-1}2^i-\sum_{i=0}^{t-1}2^i+(2+t-1)
$$
$$=2^m-1-2^t+1+t+1=2^{m}-2^t+t+1,
$$
which completes the proof.
\end{proof}


In the following, we consider an anticode formed by another modification of a simplex code and prove that the LRC obtained from this anticode attains  the bound~(\ref{eq_CMBound}) and has locality 3.

\begin{theorem} \label{thm:anticodeSm-2}
Let $\cA_{m-1}$ be the anticode with generator matrix $G_{\cA}$  given by
$$G_{\cA}=\left(
    \begin{array}{c|c}
      1&000\ldots 00  \\\hline
       \begin{array}{c}
                    0 \\
                    \vdots \\
                    0 \\
                  \end{array} &G_{m-1}
                      \\
    \end{array}
  \right),
$$
where $G_{m-1}$ is the generator matrix for the simplex code $\simplex{m-1}$.
Let $\cC$ be a code obtained by the Farrell construction based on the simplex code $\simplex{m}$ and on the anticode $\cA_{m-1}$. Then $\cC$ is a $[2^{m-1}-1,m,2^{m-2}-1]$ LRC with locality $r=3$ which attains the bound~(\ref{eq_CMBound}).
\end{theorem}
\begin{proof}
Since all the codewords in $\simplex{m-1}$ have the constant weight $2^{m-2}$, the maximum weight of $\cA_{m-1}$ is ${\delta=2^{m-2}+1}$. Then, by the Farrell construction, $\cC$ is a $[2^{m}-1-2^{m-1}, m, 2^{m-1}-2^{m-2}-1]=[2^{m-1}-1,m,2^{m-2}-1]$ code. Note that $\cC$ is the augmented simplex code $\simplex{m-1}$ with generator matrix:
$$G=\left(
    \begin{array}{c}
      111\ldots 11  \\\hline
      G_{m-1}           \\
    \end{array}
  \right).
$$
To prove that the locality is 3, we note that all the codewords in the dual code
of $\cC$ have even weight, and hence the locality $r$ is an odd number. Clearly, ${r>1}$. We construct the parity-check matrix $H$ of $\cC$ with all the rows of weight 4, and then every coordinate will be covered by a row of $H$ of weight~4, which by Lemma~\ref{rm:locality from H} implies that $r=3$. Recall that the dual code of $\simplex{m-1}$ is a $[2^{m-1},2^{m-1}-m,3]$ Hamming code~\cite{macwilliams_theory_1988}, and denote its generator matrix by $H^{m-1}$. We consider the construction of $H^{m-1}$ with all the rows of weight 3 from~\cite{cadambe_upper_2013}, where the rows have nonzero entries in the positions ${(i,2^j, i+2^j)}$, for $1\leq j\leq m-2$, $1\leq i\leq 2^j$. Let denote by $H_{i,2^j}^{m-1}$ a row of $H^{m-1}$ with nonzero entries in the positions ${(i,2^j, i+2^j)}$. The parity-check matrix $H$ will consist of $2^{m-1}-m-1$ weight-4 rows $H^{m-1}_{1,2}+H^{m-1}_{1,2^j}$, $2\leq j\leq m-2$, and $H^{m-1}_{i,2^j}+H^{m-1}_{i+1,2^j}$, $2\leq j\leq m-2$, $1\leq i\leq 2^{j}-2$.

To prove the optimality of the obtained code $\cC$ we apply the bound~(\ref{eq_CMBound}) with $t=1$ and use the Plotkin bound:
$$3+k_{opt}^{(2)}(2^{m-1}-1-4,2^{m-2}-1)\leq 3+\left\lfloor\log 2\cdot\left\lfloor\frac{2^{m-2}-1}{3}\right\rfloor\right\rfloor
$$
$$
\leq 3+\left\lfloor\log(2^{m-2}-1)\right\rfloor=3+m-3=m.$$
\end{proof}

\begin{remark}
One can check that the code $\cC$  from Thm.~\ref{thm:anticodeSm-1} attains the Griesmer bound.
\end{remark}


\subsection{LRCs based on Subspace Codes}
In this subsection we consider a construction of optimal binary LRCs with the parity-check matrix formed by the codewords of a
\emph{subspace} code. In particular, we are interested in a special kind of subspace codes, called \emph{lifted rank-metric codes}~\cite{silberstein_codes_2011}, with the trivial distance 1 and the constant dimension 2.
 More precisely, let consider a $2^{2s-4}\times (2^s-2^{s-2})$ binary matrix $H^s$ whose columns are indexed by the vectors in $\F_2^s\setminus \{00v:v\in \F_2^{s-2}\}$ and whose rows are indexed by the 2-dimensional subspaces of $\F_2^s$ contained in $\F_2^s\setminus \{00v:v\in \F_2^{s-2}\}$. Every row of $H^s$ is the incidence vector of such a 2-dimensional subspace and then has weight~3 (note that the all-zero vector is not considered).  It was proved in~\cite[Thm. 11]{silberstein_codes_2011} that the code $\cC^s$ with the parity-check matrix $H^s$  is a $[2^s-2^{s-2},s,\frac{2^s-2^{s-2}}{2}]$ $2^{s-2}$-quasi-cyclic code. Note that $H^s$ contains dependent rows.
 \begin{example} For $s=4$, the matrix $H^4$ has the following form. The four rows above this matrix
represent the vectors which index the columns of $H^4$. For example, the first row of $H^4$ corresponds to the subspace which contains the vectors $\{(0,1,0,0),(1,0,0,0),(1,1,0,0)\}$:

 \begin{footnotesize}
$$\begin{array}{c}
 \left. \begin{tabular}{c|c|c}
    \bf {0 0 0 0} & \bf {1 1 1 1} & \bf {1 1 1 1} \\
    \bf {1 1 1 1} & \bf {0 0 0 0} & \bf {1 1 1 1} \\
    \bf {0 0 1 1} & \bf {0 0 1 1} & \bf {0 0 1 1} \\
    \bf {0 1 0 1} & \bf {0 1 0 1} & \bf {0 1 0 1} \\
    \end{tabular} \right. \\
    \left(\begin{tabular}{c|c|c}\hline\hline
    \tn{1 0 0 0} & \tn{1 0 0 0} & \tn{1 0 0 0} \\
    \tn{0 1 0 0} & \tn{0 0 0 1} & \tn{0 0 1 0} \\
    \tn{0 0 1 0} & \tn{0 1 0 0} & \tn{0 0 0 1} \\
    \tn{0 0 0 1} & \tn{0 0 1 0} & \tn{0 1 0 0} \\
     \hline
    \tn{1 0 0 0} & \tn{0 1 0 0} & \tn{0 1 0 0} \\
    \tn{0 1 0 0} & \tn{0 0 1 0} & \tn{0 0 0 1} \\
    \tn{0 0 1 0} & \tn{1 0 0 0} & \tn{0 0 1 0} \\
    \tn{0 0 0 1} & \tn{0 0 0 1} & \tn{1 0 0 0} \\
     \hline
    \tn{1 0 0 0} & \tn{0 0 1 0} & \tn{0 0 1 0} \\
    \tn{0 1 0 0} & \tn{0 1 0 0} & \tn{1 0 0 0} \\
    \tn{0 0 1 0} & \tn{0 0 0 1} & \tn{0 1 0 0} \\
    \tn{0 0 0 1} & \tn{1 0 0 0} & \tn{0 0 0 1} \\
     \hline
    \tn{1 0 0 0} & \tn{0 0 0 1} & \tn{0 0 0 1} \\
    \tn{0 1 0 0} & \tn{1 0 0 0} & \tn{0 1 0 0} \\
    \tn{0 0 1 0} & \tn{0 0 1 0} & \tn{1 0 0 0} \\
    \tn{0 0 0 1} & \tn{0 1 0 0} & \tn{0 0 1 0} \\
    \end{tabular}
\right)
\end{array}$$
\end{footnotesize}
 \end{example}

 In the following we prove that this code attains the bound~(\ref{eq_CMBound}) with locality 2.

\begin{theorem} \label{thm:liftedMRd}
Let $\cC^s$ be the linear code with the parity-check matrix $H^s$ defined above. Then $\cC^s$ is a $[3\cdot2^{s-2},s,3\cdot2^{s-3}]$ LRC with locality $r=2$ which attains the bound~(\ref{eq_CMBound}).
\end{theorem}
 \begin{proof}
 The length, the dimension and the  minimum distance are proved in~\cite[Thm. 11]{silberstein_codes_2011}. Since every row in $H^s$ has weight 3, the locality is 2.
  By applying the bound~(\ref{eq_CMBound}) with $t=1$ and using the Plotkin bound we have
  $$2+k_{opt}^{(2)}(3\cdot2^{s-2}-3,3\cdot2^{s-3})
   \leq 2+\log 2\cdot\left\lfloor\frac{3\cdot2^{s-3}}{3}\right\rfloor$$
  $$=2+\log 2^{s-2}=2+s-2=s$$
  In other words, the code $\cC^s$  always attain the bound~(\ref{eq_CMBound}).
\end{proof}
\begin{remark}
One can check that the code $\cC^s$  from Thm.~\ref{thm:liftedMRd} attains the  Griesmer bound.
\end{remark}
\begin{remark}  Note that the code $\cC^s$ from Thm.~\ref{thm:liftedMRd} can be also constructed by applying the Farrell construction on a simplex code $\simplex{s}$, when using a simplex code $\simplex{s-2}$ as an anticode, as follows.  By the construction of $\cC^s$, the columns of the generator matrix for $\cC^s$ are formed by all the vectors in $\F_2^s\setminus \{00v:v\in \F_2^{s-2}\}$. Then the columns of the generator matrix for the anticode are formed by the vectors in $\{00v:v\in \F_2^{s-2}\}$. This anticode has length $2^{s-2}-1$ and maximum weight $2^{s-3}$.
\end{remark}

\section{Conclusion}
\label{sec:conclusion}
We presented a construction for four families of
binary linear optimal LRCs with locality $r=2$ and $r=3$ (see Tab.~\ref{tab_CodeExamples} for some numerical examples). This construction is based on various anticodes.  Besides the optimality with respect to the Cadambe--Mazumdar bound for a given locality, several of our families of codes fulfill the Griesmer bound with equality.

\begin{table}[htb]
\caption{Optimal Binary LRC with Locality Two and Three.}
\label{tab_CodeExamples}
\begin{center}
\begin{tabular}{lcll}
$[n,k,d]$  & Locality $r$ & Reference & \\
\hline
$[28,5,14]$ & 2 & Thm.~\ref{thm:anticode-weight2} & $m=5, s=3$ \\
$[25,5,12]$ & 2 & & $m=5, s=4$ \\
$[21,5,10]$ & 2 & & $m=5, s=5$ \\
\hline
$[60,6,30]$ & 2 & Thm.~\ref{thm:anticode-weight2}  & $m=6, s=3$ \\
$[57,6,28]$ & 2 &  & $m=6, s=4$ \\
$[53,6,26]$ & 2 &  & $m=6, s=5$ \\
\hline
$[21,5,10]$ & 2 & Thm.~\ref{thm:anticodeSm-1} & $m=5$, $t=4$ \\
$[38,6,18]$ & 2 & & $m=6$, $t=5$ \\
\hline
$[31,6,15]$ & 3 & Thm.~\ref{thm:anticodeSm-2} & $m=6$ \\
$[63,7,31]$ & 3 & & $m=7$ \\
\hline
$[24,5,12]$ & 2 & Thm.~\ref{thm:liftedMRd} & $s=5$ \\
$[48,6,24]$ & 2 & & $s=6$
\end{tabular}
\end{center}
\end{table}

\vspace{0.5cm}

\end{document}